\documentclass[preprint,3p]{elsarticle}
\usepackage{amsmath,amssymb,amsthm,mathrsfs}
\usepackage{tikz}
  \usetikzlibrary{arrows,automata,shapes}
  \usetikzlibrary{positioning}
\usepackage[linesnumbered,ruled,vlined]{algorithm2e}

\newcommand{\eps}{\varepsilon}

\DeclareMathOperator{\alp}{\rm alph}
\DeclareMathOperator{\0}{\bf 0}
\DeclareMathOperator{\1}{\bf 1}
\newcommand{\A}{\mathcal{A}}
\newcommand{\B}{\mathcal{B}}

\newtheorem{thm}{Theorem}
\newtheorem{prop}[thm]{Proposition}

\journal{}

\begin{document}
\begin{frontmatter}

\title{Separability by Piecewise Testable Languages is {\sc PTime}-Complete}

\author{Tom\'{a}\v{s} Masopust}
\ead{masopust@math.cas.cz}

\address{Institute of Mathematics, Czech Academy of Sciences, {\v Z}i{\v z}kova 22, 616 62 Brno, Czechia}

\begin{abstract}
  Piecewise testable languages form the first level of the Straubing-Th\'erien hierarchy. The membership problem for this level is decidable and testing if the language of a DFA is piecewise testable is NL-complete. The question has not yet been addressed for NFAs. We fill in this gap by showing that it is {\sc PSpace}-complete. The main result is then the lower-bound complexity of separability of regular languages by piecewise testable languages. Two regular languages are separable by a piecewise testable language if the piecewise testable language includes one of them and is disjoint from the other. For languages represented by NFAs, separability by piecewise testable languages is known to be decidable in {\sc PTime}. We show that it is {\sc PTime}-hard and that it remains {\sc PTime}-hard even for minimal DFAs. 
\end{abstract}


\end{frontmatter}

\section{Introduction}
  A regular language over $\Sigma$ is {\em piecewise testable\/} if it is a finite boolean combination of languages of the form $\Sigma^* a_1 \Sigma^* a_2 \Sigma^* \cdots \Sigma^* a_n \Sigma^*$, where $a_i\in \Sigma$ and $n\ge 0$. If $n$ is bounded by a constant $k$, then the language is called {\em $k$-piecewise testable}. Piecewise testable languages are exactly those regular languages whose syntactic monoid is $\mathcal{J}$-trivial~\cite{Simon1972}. Simon~\cite{Simon1975} provided various characterizations of piecewise testable languages, e.g., in terms of monoids or automata. These languages are of interest in many disciplines of mathematics, such as semigroup theory~\cite{Almeida2008486,AlmeidaZ-ita97,PerrinPin} for their relation to Green's relations or in logic on words~\cite{DiekertGK08} for their relation to first-order logic $\textrm{FO}[<]$ and the {\em Straubing-Th\'erien hierarchy}~\cite{Straubing81,Therien81}. 
  
  For an alphabet $\Sigma$, level 0 of the Straubing-Th\'erien hierarchy is defined as $\mathscr{L}(0)=\{\emptyset, \Sigma^*\}$. For integers $n\geq 0$, level $\mathscr{L}(n+\frac{1}{2})$ consists of all finite unions of languages $L_0 a_1 L_1 a_2 \ldots a_k L_k$ with $k\geq 0$, $L_0,\ldots, L_k\in\mathscr{L}(n)$, and $a_1,\ldots,a_k\in\Sigma$, and level $\mathscr{L}(n+1)$ consists of all finite Boolean combinations of languages from level $\mathscr{L}(n+\frac{1}{2})$. The levels of the hierarchy contain only \emph{star-free} languages~\cite{McNaughton1971}. Piecewise testable languages form the first level of the hierarchy. The hierarchy does not collapse on any level \cite{BrzozowskiK78}. In spite of a recent development~\cite{AlmeidaBKK15,Place15,PlaceZ15}, deciding whether a language belongs to level $\ell$ of the hierarchy is open for $\ell > \frac{7}{2}$. The Straubing-Th\'erien hierarchy is further closely related to the \emph{dot-depth hierarchy}~\cite{BrzozowskiK78,CohenB71,KufleitnerL12,Straubing85} and to complexity theory~\cite{Wagner04}.

  The fundamental question is how to efficiently recognize whether a given regular language is piecewise testable. Stern~\cite{Stern85a} provided a solution that was later improved by Trahtman~\cite{Trahtman2001} and Kl\'ima and Pol\'ak~\cite{KlimaP13}. Stern presented an algorithm deciding piecewise testability of a regular language represented by a DFA in time $O(n^5)$, where $n$ is the number of states of the DFA. Trahtman improved Stern's algorithm to time quadratic with respect to the number of states and linear with respect to the size of the alphabet, and Kl\'ima and Pol\'ak found an algorithm for DFAs that is quadratic with respect to the size of the alphabet and linear with respect to the number of states. Cho and Huynh~\cite{ChoH91} proved that deciding piecewise testability for DFAs is NL-complete. Event though the complexity for DFAs has intensively been investigated, a study for NFAs is missing in the literature. We fill in this gap by showing that deciding piecewise testability for NFAs is {\sc PSpace}-complete (Theorem~\ref{pspace-complete}).
  
  The knowledge of the minimal $k$ or a reasonable bound on $k$ for which a piecewise testable language is $k$-piecewise testable is of interest in several applications~\cite{MartensNNS15,HofmanM15}. The complexity of finding the minimal $k$ has been investigated in the literature~\cite{HofmanM15,KKP,KlimaP13,dlt15}. Testing whether a piecewise testable language is $k$-piecewise testable is coNP-complete for $k\ge 4$ if the language is represented as a DFA~\cite{KKP} and {\sc PSpace}-complete if the language is represented as an NFA~\cite{dlt15}. The complexity for DFAs and $k<4$ has also been discussed in detail~\cite{dlt15}. Kl\'ima and Pol\'ak~\cite{KlimaP13} further showed that the upper bound on $k$ is given by the depth of the minimal DFA. This result has recently been generalized to NFAs~\cite{ptnfas}.

  The recent interest in piecewise testable languages is mainly because of the applications of separability of regular languages by piecewise testable languages in logic on words~\cite{PlaceZ_icalp14} and in XML schema languages~\cite{icalp2013,HofmanM15,MartensNNS15}. Given two languages $K$ and $L$ and a family of languages $\mathcal{F}$, the {\em separability problem\/} asks whether there exists a language $S$ in $\mathcal{F}$ such that $S$ includes one of the languages $K$ and $L$ and is disjoint from the other. 
  Place and Zeitoun~\cite{PlaceZ_icalp14} used separability to obtain new decidability results of the membership problem for some levels of the Straubing-Th\'erien hierarchy.
  The separability problem for regular languages represented by NFAs and the family of piecewise testable languages is decidable in polynomial time with respect to both the number of states and the size of the alphabetby~\cite{icalp2013,mfcsPlaceRZ13}. 
  Separability by piecewise testable languages is of interest also outside regular languages. Although separability of context-free languages by regular languages is undecidable~\cite{Hunt82a}, separability by piecewise testable languages is decidable (even for some non-context-free languages)~\cite{CzerwinskiMRZ15}.
  Piecewise testable languages are further investigated in natural language processing~\cite{FuHT2011,Rogers:2007}, cognitive and sub-regular complexity~\cite{RogersHFHLW12}, and learning theory~\cite{GarciaR04,Kontorovich2008}. They have been extended from word languages to tree languages~\cite{Bojanczyk:2012,GarciaV90,Goubault-Larrecq16}.
  
  In this paper, we show that separability of regular languages represented as NFAs by piecewise testable languages is a {\sc PTime}-complete problem (Theorem~\ref{p-complete}) and that it remains {\sc PTime}-hard even for minimal DFAs. Consequently, the separability problem is unlikely to be solvable in logarithmic space or effectively parallelizable.

\section{Preliminaries}
  We assume that the reader is familiar with automata theory~\cite{sipser}. The cardinality of a set $A$ is denoted by $|A|$ and the power set of $A$ by $2^A$. The free monoid generated by an alphabet $\Sigma$ is denoted by $\Sigma^*$. A word over $\Sigma$ is any element of $\Sigma^*$; the empty word is denoted by $\eps$. For a word $w\in\Sigma^*$, $\alp(w)\subseteq\Sigma$ denotes the set of all symbols occurring in $w$.

  A {\em nondeterministic finite automaton\/} (NFA) is a quintuple $\A = (Q,\Sigma,\delta,Q_0,F)$, where $Q$ is the finite nonempty set of states, $\Sigma$ is the input alphabet, $Q_0\subseteq Q$ is the set of initial states, $F\subseteq Q$ is the set of accepting states, and $\delta\colon Q\times\Sigma\to 2^Q$ is the transition function extended to the domain $2^Q\times\Sigma^*$ in the usual way. The language {\em accepted\/} by $\A$ is the set $L(\A) = \{w\in\Sigma^* \mid \delta(Q_0, w) \cap F \neq\emptyset\}$.   

  A {\em path\/} $\pi$ from a state $q_0$ to a state $q_n$ under a word $a_1a_2\cdots a_{n}$, for some $n\ge 0$, is a sequence of states and input symbols $q_0, a_1, q_1, a_2, \ldots, q_{n-1}, a_{n}, q_n$ such that $q_{i+1} \in \delta(q_i,a_{i+1})$, for all $i=0,1,\ldots,n-1$. Path $\pi$ is {\em accepting\/} if $q_0\in Q_0$ and $q_n\in F$. We write $q_0 \xrightarrow{a_1a_2\cdots a_{n}} q_{n}$ to denote that there is a path from $q_0$ to $q_n$ under the word $a_1a_2\cdots a_{n}$. 
  
  We says that $\A$ has a {\em cycle over an alphabet $\Gamma\subseteq\Sigma$\/} if there is a state $q$ in $\A$ and a word $w$ over $\Sigma$ such that $q\xrightarrow{\,w\,} q$ and $\alp(w)=\Gamma$.

  The NFA $\A$ is {\em deterministic\/} (DFA) if $|Q_0|=1$ and $|\delta(q,a)|=1$ for every $q \in Q$ and $a \in \Sigma$. Although we define DFAs as complete, we mostly depict only the most important transitions in our illustrations. The reader can easily complete such an incomplete DFA.

  Let $K$ and $L$ be languages. A language $S$ \emph{separates $K$ from $L$\/} if $S$ contains $K$ and does not intersect $L$. Languages $K$ and $L$ are \emph{separable by a family of languages $\mathcal{F}$\/} if there exists a language $S$ in $\mathcal{F}$ that separates $K$ from $L$ or $L$ from $K$.

\section{Piecewise Testability for NFAs}
  Given an NFA $\A$ over an alphabet $\Sigma$, the {\em piecewise-testability problem\/} asks whether the language $L(\A)$ is piecewise testable. Although the membership in {\sc PSpace} follows basically from the result by Cho and Huynh~\cite{ChoH91}, we prefer to provide the proof here for two reasons: (i) we would like to provide unfamiliar readers with a method to recognize whether a regular language is piecewise testable, (ii) Cho and Huynh assume that the input is a minimal DFA, hence it is necessary to extend their algorithm with a non-equivalence check. We use the following characterization in our proof.
  
  \begin{prop}[Simon~\cite{Simon1975}, Cho and Huynh~{\cite[Proposition 2.3(b)]{ChoH91}}]\label{ChoHcharacterization}
    A regular language $L$ is not piecewise testable if and only if the minimal DFA for $L$ either 
    \begin{enumerate}
      \item contains a nontrivial (non-self-loop) cycle or 
      \item there are three distinct states $p$, $q$, $q'$ such that $q$ and $q'$ are reachable from $p$ by words over the symbols that form self-loops on both $q$ and $q'$; formally, there are paths $p \xrightarrow{w} q$ and $p \xrightarrow{w'} q'$ in the DFA with $w,w' \in \Sigma(q) \cap \Sigma(q')$, where $\Sigma(q)=\{ a\in\Sigma \mid q\xrightarrow{a} q\}$.
    \end{enumerate}
  \end{prop}

  We now prove the first result of this paper.
  \begin{thm}\label{pspace-complete}
    The piecewise-testability problem for NFAs is {\sc PSpace}-complete.
  \end{thm}
  \begin{proof}
    To prove that piecewise testability is in {\sc PSpace}, let $\A=(Q,\Sigma,\delta,Q_0,F)$ be an NFA. Since $\A$ is nondeterministic, we cannot directly use the algorithm of Cho and Huynh~\cite{ChoH91}. Instead, we consider the DFA $\A'$ obtained from $\A$ by the standard subset construction where the states of $\A'$ are subsets of states of $\A$. We now need to modify Cho and Huynh's algorithm to check whether the guessed states are distinguishable.
    \begin{algorithm}
      \DontPrintSemicolon
      \caption{Non-piecewise testability (symbol $\rightsquigarrow$ stands for reachability)}
      \label{alg}\label{algpspace}
      \SetKwInOut{Input}{Input}\SetKwInOut{Output}{Output}
        \Input{An NFA  $\A=(Q,\Sigma,\delta,Q_0,F)$}
        \Output{{\tt true} if and only if $L(\A)$ is not piecewise testable}
        Guess states $X,Y\subseteq Q$ of $\A'$ \tcp*{Verify property (1)}
        \lIf {$Q_0 \rightsquigarrow X \rightsquigarrow Y \rightsquigarrow X$}
             {go to line~\ref{noneq}}
        Guess states $P,X,Y\subseteq Q$ of $\A'$ \tcp*{Verify property (2)}
        Check that $Q_0\rightsquigarrow P$, $Q_0\rightsquigarrow X$, and $Q_0\rightsquigarrow Y$\;
        $s_1 := P$, $s_2 := P$\;
        \Repeat {$s_1 = X$ and $s_2 = Y$}{
          guess $a,b\in \Sigma(X)\cap \Sigma(Y)$\;
          $s_1 := \delta(s_1,a)$,
          $s_2 := \delta(s_2,b)$\;
        }
        Guess states $X',Y'$ of $\A'$ s.\,t. $X'\cap F \neq\emptyset$ and $Y'\cap F=\emptyset$; 
          \tcp*{Non-equiv.\ check of $X$ and $Y$}
          \label{noneq}
        $s_1 := X$, $s_2 := Y$\; 
        \Repeat {$s_1 = X'$ and $s_2 = Y'$}{
          guess $a\in\Sigma$\;
          $s_1 := \delta(s_1,a)$,
          $s_2 := \delta(s_2,a)$\;
        }
        \Return {\tt true}
    \end{algorithm}
    For a set of states $X\subseteq Q$, let $\Sigma(X)=\{ a \in \Sigma \mid X \xrightarrow{a} X\}$. The entire algorithm is presented as Algorithm~\ref{algpspace}. 
    
    In line~1 the algorithm guesses two states, $X$ and $Y$, of $\A'$ that are verified to be reachable and in a cycle in line 2. If so, it is verified in lines 10--15 that the states $X$ and $Y$ are not equivalent in $\A'$. If there is no nontrivial cycle in $\A'$ or the guess in line~1 fails, property (2) of Proposition~\ref{ChoHcharacterization} is verified in lines 3--9, and the guessed states $X$ and $Y$ are checked to be non-equivalent in lines 10--15. Notice that in lines 6--9, the algorithm verifies that the states $X$ and $Y$ are reachable from a state $P$ by paths of the same length rather than by paths of different lengths. This is not a problem because line~7 considers only symbols from $\Sigma(X)\cap \Sigma(Y)$. If $\A'$ reaches $X$ under $\Sigma(X)\cap \Sigma(Y)$, it stays in $X$ under those symbols (and analogously for $Y$). Thus, under $\Sigma(X)\cap \Sigma(Y)$, the states $X$ and $Y$ are reachable from state $P$ by paths of different lengths if and only if they are reachable by paths of the same length. The algorithm is in {\sc NPSpace $ = $ PSpace}~\cite{Savitch1970} and returns a positive answer if and only if $\A$ does not accept a piecewise testable language. Since {\sc PSpace} is closed under complement~\cite{Immerman88,Szelepcsenyi87}, piecewise testability is in {\sc PSpace}.

    {\sc PSpace}-hardness follows from a result by Hunt III and Rosenkrantz~\cite{HuntR78}, who have shown that a property $\mathbb{P}$ of languages over the alphabet $\{0,1\}$ such that (i) $\mathbb{P}(\{0,1\}^*)$ is true and (ii) there exists a regular language that is not expressible as a quotient $x\backslash L=\{ w \mid xw \in L\}$, for some $L$ for which $\mathbb{P}(L)$ is true, is as hard as to decide ``$=\{0,1\}^*$''. Since piecewise testability is such a property (piecewise testable languages are closed under quotient) and universality is {\sc PSpace}-hard for NFAs, the result implies that piecewise testability for NFAs is {\sc PSpace}-hard.
  \end{proof}

\section{Separability of Regular Languages by Piecewise Testable Languages}
  We now show that separability of regular languages by piecewise testable languages is {\sc PTime}-complete. Since the membership in {\sc PTime} is known~\cite{icalp2013,mfcsPlaceRZ13}, we prove {\sc PTime}-hardness by constructing a log-space reduction from the {\sc PTime}-complete monotone circuit value problem~\cite{limits}. 

  The {\em monotone circuit value problem\/} consists of a set of boolean variables $g_1$, $g_2$, \ldots, $g_n$ called {\em gates}, whose values are defined recursively by equalities of the forms $g_i = \0$ (then $g_i$ is called a $\0$-gate), $g_i = \1$ ($\1$-gate), $g_i = g_j \wedge g_k$ ($\wedge$-gate), or $g_i = g_j \vee g_k$ ($\vee$-gate), where $j, k < i$. Here $\0$ and $\1$ are symbols representing the boolean values. The aim is to compute the value of $g_n$. 
  
  A word $a_1 a_2 \cdots a_n$ with $a_i\in\Sigma$ is a subsequence of a word $w$ if $w \in \Sigma^* a_1 \Sigma^* a_2 \Sigma^* \cdots \Sigma^* a_n \Sigma^*$. For languages $K$ and $L$, a sequence $(w_i)_{i=1}^{r}$ of words is a {\em tower\/} between $K$ and $L$ if $w_1 \in K \cup L$ and, for all $i = 1, 2, \ldots, r-1$,
    $w_i$ is a subsequence of $w_{i+1}$,
    $w_i \in K$ implies $w_{i+1} \in L$, and
    $w_i \in L$ implies $w_{i+1} \in K$.
  The number of words in the sequence is the {\em height\/} of the tower; the height may be infinite. Languages $K$ and $L$ are not required to be disjoint, but a $w \in K \cap L$ implies an infinite tower $w,w,\ldots$ between $K$ and $L$. 
  
  Our proof is based on the fact that non-separability of languages $K$ and $L$ by a piecewise testable language is equivalent to the existence of an infinite tower between the languages $K$ and $L$~\cite{icalp2013}.

  \begin{thm}\label{p-complete}
    Deciding separability of regular languages represented as NFAs by piecewise testable languages is {\sc PTime}-complete. It remains {\sc PTime}-hard even for minimal DFAs.
  \end{thm}
  \begin{proof}
    The membership in {\sc PTime} was independently shown by Czerwi\'nski et al.~\cite{icalp2013} and Place et al.~\cite{mfcsPlaceRZ13}. 

    We prove {\sc PTime}-hardness by reduction from the monotone circuit value problem (MCVP). Given an instance $g_1, g_2,\ldots, g_n$ of MCVP, we construct two minimal DFAs $\A$ and $\B$ using a log-space reduction and prove that there exists an infinite tower between their languages if and only if the circuit evaluates gate $g_n$ to~$\1$. The theorem then follows from the fact that non-separability of two regular languages by a piecewise testable language is equivalent to the existence of an infinite tower~\cite{icalp2013}.
    
    Let $f(i)$ be the element of $\{\wedge,\vee,\0,\1\}$ such that $g_i$ is an $f(i)$-gate. For every $\wedge$-gate and $\vee$-gate, we set $\ell(i)$ and $r(i)$ to be the indices such that $g_i=g_{\ell(i)}f(i)g_{r(i)}$ is the defining equality of $g_i$. If $g_i$ is a $\0$-gate, we set $f(i)=\ell(i)=r(i)=\0$, and if $g_i$ is a $\1$-gate, we set $f(i) = \ell(i) = r(i) = \1$.

    We first construct an automaton $\A'=(Q_{\A'},\Sigma,\delta_{\A'},s,F_{\A'})$ with states $Q_{\A'}= \{s,\0,\1,1,2,\dots,n\}$, the input alphabet $\Sigma=\{x,y\}\cup\{a_i,b_i \mid i=1,\ldots,n\}$, and accepting states $F_{\A'}=\{\0,\1\}$. The initial state of $\A'$ is $s$ and the transition function $\delta_{\A'}$ is defined by $\delta_{\A'}(i,a_i)=\ell(i)$ and $\delta_{\A'}(i,b_i)=r(i)$. In addition, there are two special transitions $\delta_{\A'}(s,x)=n$ and $\delta_{\A'}(\1,y)=s$. 
    
    To construct automaton $\B=(Q_{\B},\Sigma,\delta_{\B},q,F_{\B})$, let $Q_{\B}= \{q,t\}\cup\{i\mid f(i)=\wedge\}$ and $F_{\B}=\{q\}$, where $q$ is also the initial state of $\B$. If $f(i)=\vee$ or $f(i)=\1$, we define $\delta_{\B}(t,a_{i})=\delta_{\B}(t,b_{i})=t$. If $f(i)=\wedge$, we define $\delta_{\B}(t,a_{i})=i$ and $\delta_{\B}(i,b_{i})=t$. Finally, we define $\delta_{\B}(q,x)=t$ and $\delta_{\B}(t,y)=q$. 
    
    All undefined transitions go to the unique sink states of the respective automata. The automata $\A'$ and $\B$ can be constructed from $g_1,\ldots,g_n$ in logarithmic space. An example of the construction for the circuit $g_1=\0$, $g_2=\1$, $g_3=g_1\wedge g_2$, $g_4= g_3 \vee g_3$ is illustrated in Figure~\ref{figPcomp}.

    The languages $L(\A')$ and $L(\B)$ are disjoint, the automata $\A'$ and $\B$ are deterministic, and $\B$ is minimal. However, automaton $\A'$ need not be minimal because the circuit may contain gates that do not contribute to the definition of the value of $g_n$. We therefore define a minimal deterministic automaton $\A$ by adding new transitions into $\A'$, each under a fresh symbol, from state $s$ to each of the states $1,2,\dots,n-1$, from each of the states $1,2,\dots,n$ to state $\0$, and from state $\0$ to state $\1$. This can again be done in logarithmic space. No new transition is defined in $\B$. 

    \begin{figure}
      \centering
      \begin{tikzpicture}[auto,baseline,->,>=stealth,shorten >=1pt,node distance=2cm,
        state/.style={circle,minimum size=1mm,very thin,draw=black,initial text=},
        bigloop/.style={shift={(0,0.01)},text width=1.6cm,align=center},
        bigloopd/.style={shift={(0,-0.01)},text width=1.6cm,align=center}]
        \node[state,initial]    (1) {$s$};
        \node[state]            (2) [right of=1] {$4$};
        \node[state]            (3) [right of=2] {$3$};
        \node[state]            (5) [right of=3] {$1$};
        \node[state]            (4) [below=.5cm of 5] {$2$};
        \node[state,accepting]  (6) [right of=4] {$\1$};
        \node[state,accepting]  (7) [right of=5] {$\0$};
        \path
          (1) edge node{$x$} (2)
          (2) edge node{$a_{4},b_4$} (3)
          (3) edge node{$b_3$} (4)
              edge node{$a_3$} (5)
          (6) edge[bend left=34] node{$y$} (1)
          (5) edge node{$a_1,b_1$} (7)
          (4) edge node{$a_2,b_2$} (6) ;
      \end{tikzpicture}
      \qquad
      \begin{tikzpicture}[auto,baseline,->,>=stealth,shorten >=1pt,node distance=2cm,
        state/.style={circle,minimum size=1mm,very thin,draw=black,initial text=},
        bigloop/.style={shift={(0,0.01)},text width=1.6cm,align=center},
        bigloopd/.style={shift={(0,-0.01)},text width=1.6cm,align=center}]
        \node[state,initial,accepting]    (1) {$q$};
        \node[state]                      (2) [right of=1] {$t$};
        \node[state]                      (3) [right of=2] {$3$};
        \path
          (1) edge[bend left] node{$x$} (2)
          (2) edge[loop above] node[bigloop]{$a_2,b_2$\\$a_4,b_4$} (2)
          (2) edge[bend left] node{$a_3$} (3)
          (3) edge[bend left] node{$b_3$} (2)
          (2) edge[bend left] node{$y$} (1) ;
      \end{tikzpicture}
      \caption{Automata $\A'$ and $\B$ for the circuit $g_1=\0$, $g_2=\1$, $g_3=g_1\wedge g_2$, $g_4= g_3 \vee g_3$.}
      \label{figPcomp}
    \end{figure}
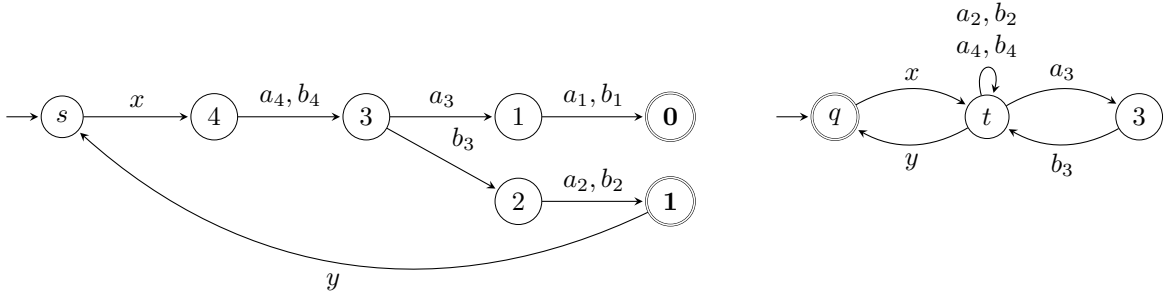

    Since the language of $\B$ is over $\Sigma$, the symbols of $\A$ not belonging to $\Sigma$ have no effect on the existence of an infinite tower between $L(\A)$ and $L(\B)$. Namely, there exists an infinite tower between the languages $L(\A)$ and $L(\B)$ if and only if there exists an infinite tower between $L(\A')$ and $L(\B)$. It is therefore sufficient to prove that the circuit evaluates gate $g_n$ to $\1$ if and only if there is an infinite tower between the languages $L(\A')$ and $L(\B)$. 
    
    The intuition behind the construction is that the symbols of an infinite tower with unbounded number of occurrences correspond to gates that evaluate to $\1$ to satisfy $g_n$, and that the non-existence of an infinite tower implies the existence of a symbol with bounded number of occurrences in $\A'$ that appears in a non-trivial cycle of the form $a_jb_j$ in $\B$. Such a state corresponds to an $\land$-gate, $g_j$, which cannot be satisfied and causes that $g_n$ evaluates to $\0$ (cf. symbol $a_3$ in Figure~\ref{figPcomp}).
    
    If there are no $\land$-gates, $g_n$ is satisfied if and only if state $\1$ is reachable from state $n$ in $\A'$. Let $w$ be a word under which state $\1$ is reachable from state $n$. Then $xw\in L(\A')$, $xwy\in L(\B)$, $xwyxw\in L(\A'),\ldots$ is an infinite tower between $L(\A')$ and $L(\B)$. If state $\1$ is not reachable from state $n$ in $\A'$, then the language $L(\A')$ is finite and there is indeed no infinite tower between $L(\A')$ and $L(\B)$.
    
    The problem with $\land$-gates is how to ensure that both children of an $\land$-gate $g_j$ are satisfied. To this aim, we use the nontrivial cycle under $a_jb_j$ in $\B$, which enforces that both $a_j$ and $b_j$ appear in the words of an infinite tower. Speaking intuitively, automata $\A'$ and $\B$ encode the satisfiability check of $g_j$ (see state $g_3$ in Figure~\ref{figPcomp}) in the following way. Automaton $\A'$ checks reachability of state $\1$ from state $j$ under a word in $a_j\Sigma^* \cup b_j\Sigma^*$ and automaton $\B$ ensures that $a_j$ appears in a word in $L(\B)$ if and only if $b_j$ does. The main idea now is that if there is an infinite tower $(w_i)_{i=1}^{\infty}$ and $a_j$ appears in a word $w_i \in L(\A')$, then both $a_j$ and $b_j$ appear in $w_{i+1}\in L(\B)$. By the construction of $\A'$, symbol $x$ appears between any two occurrences of $a_j$ and $b_j$, hence $\B$ increases the number of occurrences of $a_j$ and $b_j$ in the words of the tower as the height grows. Since the tower is infinite, the number of their occurrences is unbounded. However, to read an unbounded number of $a_j$ and $b_j$ in $\A'$ requires that there is a path from state $j$ to state $\1$ under a word in $a_j\Sigma^*$ as well as under a word in $b_j\Sigma^*$, which (using inductively the same argument for other $\land$-gates) is possible only if $g_j$ is satisfied.
    In Figure~\ref{figPcomp}, the words of $L(\A')$ contain at most one occurrence of $a_3$, whereas those of $L(\B)$ require unbounded number of occurrences of $a_3$. Thus, there is no infinite tower between the languages of Figure~\ref{figPcomp}.

    We now formally prove that the circuit evaluates gate $g_n$ to $\1$ if and only if there is an infinite tower between the languages $L(\A')$ and $L(\B)$. The dependence between the gates $g_1,g_2,\ldots,g_n$ can be depicted as a directed acyclic graph $G=(\{1,2,\ldots,n\},E)$, where $E$ is defined as $\delta_{\A'}$ without the labels, multiplicities and states $s,\0,\1$. We say that $i$ is {\em accessible\/} from $j$ if there is a path from $j$ to $i$ in $G$.

    (Only if) Assume that $g_n$ is evaluated to $\1$. We construct an alphabet $\Gamma$, $\{x,y\} \subseteq \Gamma \subseteq \Sigma$, under which both automata $\A'$ and $\B$ have a cycle containing the initial and an accepting state. These cycles then imply the existence of an infinite tower between the languages $L(\A')$ and $L(\B)$. Symbol $a_i$ belongs to $\Gamma$ if and only if $g_i$ is evaluated to $\1$, $i$ is accessible from $n$, and either $\ell(i)=\1$ or $g_{\ell(i)}$ is evaluated to $\1$. Similarly, $b_i$ belongs to $\Gamma$ if and only if $i$ is accessible from $n$, $g_i$ is evaluated to $\1$, and either $r(i)=\1$ or $g_{r(i)}$ is evaluated to $\1$.  It is not hard to observe that each transition labeled by a symbol $a_i$ or $b_i$ from $\Gamma$ is part of a path from $n$ to $\1$ in $\A'$, hence it appears on a cycle in $\A'$ from the initial state $s$ back to state $s$ through the accepting state $\1$. Moreover, the definition of $\wedge$ implies that $a_i\in \Gamma$ if and only if $b_i\in \Gamma$ for each $i=1,2,\dots, n$ such that $f(i)=\wedge$. Notice that $\B$ has a cycle from $q$ to $q$ labeled by $xa_ib_iy$ for each $i=1,2,\dots,n$ with $f(i)\neq\0$. Therefore, both automata $\A'$ and $\B$ have a cycle over the alphabet $\Gamma$ containing the initial and accepting states. The existence of an infinite tower follows.
    
    (If) Assume that there exists an infinite tower $(w_i)_{i=1}^{\infty}$ between $L(\A')$ and $L(\B)$, and, for the sake of contradiction, assume that $g_n$ is evaluated to $\0$. Note that any path from $i$ to $\1$ in $\A'$, where $g_i$ is evaluated to $\0$, must contain a state corresponding to an $\wedge$-gate that is evaluated to $\0$. In particular, this applies to any path in $\A'$ accepting a word of the infinite tower of length at least $n+2$, since such a path contains a subpath from $n$ to $\1$. Let $j$ denote the smallest positive integer such that $f(j)=\wedge$, gate $g_j$ is evaluated to $\0$, and $a_j$ or $b_j$ is in $\cup_{i=1}^\infty\alp(w_i)$. The construction of $\B$ implies that both $a_j$ and $b_j$ are in $\cup_{i=1}^\infty\alp(w_i)$ because of the nontrivial cycle $a_jb_j$. Since $g_j$ is evaluated to $\0$, there exists $c\in\{a,b\}$ such that the transition from $j$ under $c_j$ leads to  a state $\sigma$, where either $\sigma=\0$ or $\sigma<j$ and $g_{\sigma}$ is evaluated to $\0$. Consider a word $w_i\in L(\A')$ of the infinite tower containing $c_j$. If $w_i$ is accepted in $\1$, then the accepting path contains a subpath from $\sigma$ to $\1$, which yields a contradiction with the minimality of $j$. Therefore, $w_i$ is accepted in $\0$. However, no symbol of a transition to state $\0$ appears in a word accepted by $\B$ (cf. the symbols $a_1$ and $b_1$ in Figure~\ref{figPcomp}), a contradiction again.
  \end{proof}

\subsubsection*{Acknowledgements}
  The author is very grateful to {\v S}t{\v e}p{\' a}n Holub for his comments on the preliminary version of this work, and to an anonymous reviewer for valuable comments. The research was supported by the German Research Foundation (DFG) in Emmy Noether grant KR~4381/1-1 (DIAMOND).

\bibliographystyle{elsarticle-harv}
\bibliography{masopust}

\begin{thebibliography}{45}
\expandafter\ifx\csname natexlab\endcsname\relax\def\natexlab#1{#1}\fi
\expandafter\ifx\csname url\endcsname\relax
  \def\url#1{\texttt{#1}}\fi
\expandafter\ifx\csname urlprefix\endcsname\relax\def\urlprefix{URL }\fi

\bibitem[{Almeida et~al.(2015)Almeida, Barto{\v{n}}ov{\'{a}}, Kl{\'{\i}}ma, and
  Kunc}]{AlmeidaBKK15}
Almeida, J., Barto{\v{n}}ov{\'{a}}, J., Kl{\'{\i}}ma, O., Kunc, M., 2015. On
  decidability of intermediate levels of concatenation hierarchies. In:
  Developments in Language Theory. Vol. 9168 of LNCS. Springer, pp. 58--70.

\bibitem[{Almeida et~al.(2008)Almeida, Costa, and Zeitoun}]{Almeida2008486}
Almeida, J., Costa, J.~C., Zeitoun, M., 2008. Pointlike sets with respect to
  {$\mathcal{R}$} and {$\mathcal{J}$}. Journal of Pure and Applied Algebra
  212~(3), 486--499.

\bibitem[{Almeida and Zeitoun(1997)}]{AlmeidaZ-ita97}
Almeida, J., Zeitoun, M., 1997. The pseudovariety {$\mathcal{J}$} is
  hyperdecidable. RAIRO -- Theoretical Informatics and Applications 31~(5),
  457--482.

\bibitem[{Bojanczyk et~al.(2012)Bojanczyk, Segoufin, and
  Straubing}]{Bojanczyk:2012}
Bojanczyk, M., Segoufin, L., Straubing, H., 2012. Piecewise testable tree
  languages. Logical Methods in Computer Science 8~(3).

\bibitem[{Brzozowski and Knast(1978)}]{BrzozowskiK78}
Brzozowski, J.~A., Knast, R., 1978. The dot-depth hierarchy of star-free
  languages is infinite. Journal of Computer and System Sciences 16~(1),
  37--55.

\bibitem[{Cho and Huynh(1991)}]{ChoH91}
Cho, S., Huynh, D.~T., 1991. Finite-automaton aperiodicity is {PS}{\sc
  pace}-complete. Theoretical Computer Science 88~(1), 99--116.

\bibitem[{Cohen and Brzozowski(1971)}]{CohenB71}
Cohen, R.~S., Brzozowski, J.~A., 1971. Dot-depth of star-free events. Journal
  of Computer and System Sciences 5~(1), 1--16.

\bibitem[{Czerwi{\'n}ski et~al.(2013)Czerwi{\'n}ski, Martens, and
  Masopust}]{icalp2013}
Czerwi{\'n}ski, W., Martens, W., Masopust, T., 2013. Efficient separability of
  regular languages by subsequences and suffixes. In: International Colloquium
  on Automata, Languages and Programming. Vol. 7966 of LNCS. Springer, pp.
  150--161.

\bibitem[{Czerwi{\'n}ski et~al.(2015)Czerwi{\'n}ski, Martens, van Rooijen, and
  Zeitoun}]{CzerwinskiMRZ15}
Czerwi{\'n}ski, W., Martens, W., van Rooijen, L., Zeitoun, M., 2015. A note on
  decidable separability by piecewise testable languages. In: International
  Symposium on Fundamentals of Computation Theory. Vol. 9210 of LNCS. Springer,
  pp. 173--185, and its extended version https://arxiv.org/abs/1410.1042 with
  G. Zetzsche.

\bibitem[{Diekert et~al.(2008)Diekert, Gastin, and Kufleitner}]{DiekertGK08}
Diekert, V., Gastin, P., Kufleitner, M., 2008. A survey on small fragments of
  first-order logic over finite words. International Journal of Foundations of
  Computer Science 19~(3), 513--548.

\bibitem[{Fu et~al.(2011)Fu, Heinz, and Tanner}]{FuHT2011}
Fu, J., Heinz, J., Tanner, H.~G., 2011. An algebraic characterization of
  strictly piecewise languages. In: Theory and Applications of Models of
  Computation. Vol. 6648 of LNCS. Springer, pp. 252--263.

\bibitem[{Garc{\'{\i}}a and Ruiz(2004)}]{GarciaR04}
Garc{\'{\i}}a, P., Ruiz, J., 2004. Learning $k$-testable and $k$-piecewise
  testable languages from positive data. Grammars 7, 125--140.

\bibitem[{Garc{\'{\i}}a and Vidal(1990)}]{GarciaV90}
Garc{\'{\i}}a, P., Vidal, E., 1990. Inference of $k$-testable languages in the
  strict sense and application to syntactic pattern recognition. {IEEE}
  Transactions on Pattern Analysis and Machine Intelligence 12~(9), 920--925.

\bibitem[{Goubault{-}Larrecq and Schmitz(2016)}]{Goubault-Larrecq16}
Goubault{-}Larrecq, J., Schmitz, S., 2016. Deciding piecewise testable
  separability for regular tree languages. In: International Colloquium on
  Automata, Languages, and Programming. Vol.~55 of LIPIcs. Schloss Dagstuhl -
  Leibniz-Zentrum fuer Informatik, pp. 97:1--97:15.

\bibitem[{Greenlaw et~al.(1995)Greenlaw, Hoover, and Ruzzo}]{limits}
Greenlaw, R., Hoover, H.~J., Ruzzo, W.~L., 1995. Limits to Parallel
  Computation: {P}-Completeness Theory. Oxford University Press.

\bibitem[{Hofman and Martens(2015)}]{HofmanM15}
Hofman, P., Martens, W., 2015. Separability by short subsequences and subwords.
  In: International Conference on Database Theory. Vol.~31 of LIPIcs. Schloss
  Dagstuhl - Leibniz-Zentrum fuer Informatik, pp. 230--246.

\bibitem[{{Hunt III}(1982)}]{Hunt82a}
{Hunt III}, H.~B., 1982. On the decidability of grammar problems. Journal of
  the ACM 29~(2), 429--447.

\bibitem[{{Hunt III} and Rosenkrantz(1978)}]{HuntR78}
{Hunt III}, H.~B., Rosenkrantz, D.~J., 1978. Computational parallels between
  the regular and context-free languages. {SIAM} Journal on Computing 7~(1),
  99--114.

\bibitem[{Immerman(1988)}]{Immerman88}
Immerman, N., 1988. Nondeterministic space is closed under complementation.
  {SIAM} Journal on Computing 17, 935--938.

\bibitem[{Kl{\'{\i}}ma et~al.(2014)Kl{\'{\i}}ma, Kunc, and Pol{\'{a}}k}]{KKP}
Kl{\'{\i}}ma, O., Kunc, M., Pol{\'{a}}k, L., 2014. Deciding $k$-piecewise
  testability, manuscript.

\bibitem[{Kl\'{\i}ma and Pol{\'a}k(2013)}]{KlimaP13}
Kl\'{\i}ma, O., Pol{\'a}k, L., 2013. Alternative automata characterization of
  piecewise testable languages. In: Developments in Language Theory. Vol. 7907
  of LNCS. Springer, pp. 289--300.

\bibitem[{Kontorovich et~al.(2008)Kontorovich, Cortes, and
  Mohri}]{Kontorovich2008}
Kontorovich, L., Cortes, C., Mohri, M., 2008. Kernel methods for learning
  languages. Theoretical Computer Science 405~(3), 223--236.

\bibitem[{Kufleitner and Lauser(2012)}]{KufleitnerL12}
Kufleitner, M., Lauser, A., 2012. Around dot-depth one. International Journal
  of Foundations of Computer Science 23~(6), 1323--1340.

\bibitem[{Martens et~al.(2015)Martens, Neven, Niewerth, and
  Schwentick}]{MartensNNS15}
Martens, W., Neven, F., Niewerth, M., Schwentick, T., 2015. Bonxai: Combining
  the simplicity of {DTD} with the expressiveness of {XML} schema. In:
  Principles of Database Systems. pp. 145--156.

\bibitem[{Masopust(2016)}]{ptnfas}
Masopust, T., 2016. Piecewise testable languages and nondeterministic automata.
  In: Mathematical Foundations of Computer Science. Vol.~58 of LIPIcs. Schloss
  Dagstuhl - Leibniz-Zentrum fuer Informatik, pp. 67:1--67:14.

\bibitem[{Masopust and Thomazo(2017)}]{dlt15}
Masopust, T., Thomazo, M., 2017. On boolean combinations forming piecewise
  testable languages. Theoretical Computer Science 682, 165--179.

\bibitem[{McNaughton and Papert(1971)}]{McNaughton1971}
McNaughton, R., Papert, S.~A., 1971. Counter-Free Automata. The MIT Press.

\bibitem[{Perrin and Pin(2004)}]{PerrinPin}
Perrin, D., Pin, J.-E., 2004. Infinite words: {A}utomata, semigroups, logic and
  games. Vol. 141 of Pure and Applied Mathematics. Elsevier, pp. 133--185.

\bibitem[{Place(2015)}]{Place15}
Place, T., 2015. Separating regular languages with two quantifiers
  alternations. In: {ACM/IEEE} Symposium on Logic in Computer Science. {IEEE}
  Computer Society, pp. 202--213.

\bibitem[{Place et~al.(2013)Place, van Rooijen, and Zeitoun}]{mfcsPlaceRZ13}
Place, T., van Rooijen, L., Zeitoun, M., 2013. Separating regular languages by
  piecewise testable and unambiguous languages. In: Mathematical Foundations of
  Computer Science. Vol. 8087 of LNCS. Springer, pp. 729--740.

\bibitem[{Place and Zeitoun(2014)}]{PlaceZ_icalp14}
Place, T., Zeitoun, M., 2014. Going higher in the first-order quantifier
  alternation hierarchy on words. In: International Colloquium on Automata,
  Languages and Programming. Vol. 8573 of LNCS. Springer, pp. 342--353.

\bibitem[{Place and Zeitoun(2015)}]{PlaceZ15}
Place, T., Zeitoun, M., 2015. Separation and the successor relation. In:
  Symposium on Theoretical Aspects of Computer Science. Vol.~30 of LIPIcs.
  Schloss Dagstuhl - Leibniz-Zentrum fuer Informatik, pp. 662--675.

\bibitem[{Rogers et~al.(2010)Rogers, Heinz, Bailey, Edlefsen, Visscher,
  Wellcome, and Wibel}]{Rogers:2007}
Rogers, J., Heinz, J., Bailey, G., Edlefsen, M., Visscher, M., Wellcome, D.,
  Wibel, S., 2010. On languages piecewise testable in the strict sense. In: The
  Mathematics of Language. Vol. 6149 of LNAI. Springer, pp. 255--265.

\bibitem[{Rogers et~al.(2013)Rogers, Heinz, Fero, Hurst, Lambert, and
  Wibel}]{RogersHFHLW12}
Rogers, J., Heinz, J., Fero, M., Hurst, J., Lambert, D., Wibel, S., 2013.
  Cognitive and sub-regular complexity. In: Formal Grammar. Vol. 8036 of LNCS.
  Springer, pp. 90--108.

\bibitem[{Savitch(1970)}]{Savitch1970}
Savitch, W.~J., 1970. Relationships between nondeterministic and deterministic
  tape complexities. Journal of Computer and System Sciences 4~(2), 177--192.

\bibitem[{Simon(1972)}]{Simon1972}
Simon, I., 1972. Hierarchies of events with dot-depth one. Ph.D. thesis,
  University of Waterloo, Canada.

\bibitem[{Simon(1975)}]{Simon1975}
Simon, I., 1975. Piecewise testable events. In: GI Conference on Automata
  Theory and Formal Languages. Springer, pp. 214--222.

\bibitem[{Sipser(2006)}]{sipser}
Sipser, M., 2006. Introduction to the theory of computation, 2nd Edition.
  Thompson Course Technology.

\bibitem[{Stern(1985)}]{Stern85a}
Stern, J., 1985. Complexity of some problems from the theory of automata.
  Information and Control 66~(3), 163--176.

\bibitem[{Straubing(1981)}]{Straubing81}
Straubing, H., 1981. A generalization of the {S}ch{\"{u}}tzenberger product of
  finite monoids. Theoretical Computer Science 13, 137--150.

\bibitem[{Straubing(1985)}]{Straubing85}
Straubing, H., 1985. Finite semigroup varieties of the form {{\bf V}*{\bf D}}.
  Journal of Pure and Applied Algebra 36, 53--94.

\bibitem[{Szelepcs{\'{e}}nyi(1988)}]{Szelepcsenyi87}
Szelepcs{\'{e}}nyi, R., 1988. The method of forced enumeration for
  nondeterministic automata. Acta Informatica 26, 279--284.

\bibitem[{Th{\'{e}}rien(1981)}]{Therien81}
Th{\'{e}}rien, D., 1981. Classification of finite monoids: {T}he language
  approach. Theoretical Computer Science 14, 195--208.

\bibitem[{Trahtman(2001)}]{Trahtman2001}
Trahtman, A.~N., 2001. Piecewise and local threshold testability of {DFA}. In:
  International Symposium on Fundamentals of Computation Theory. Vol. 2138 of
  LNCS. Springer, pp. 347--358.

\bibitem[{Wagner(2004)}]{Wagner04}
Wagner, K.~W., 2004. Leaf language classes. In: Machines, Computations, and
  Universality. Vol. 3354 of LNCS. Springer, pp. 60--81.

\end{thebibliography}
 
\end{document}